\documentclass[twocolumn,showpacs,preprintnumbers,amsmath,amssymb]{revtex4}

\usepackage{float}
\usepackage{graphicx}

\usepackage{multirow}

\usepackage{braket}
\usepackage{graphicx}
\usepackage{dcolumn}
\usepackage{bm}
\usepackage{amsthm}
\usepackage{amsmath}
\usepackage{amssymb}
\newtheorem{theorem}{Theorem}

\usepackage{algorithm}
\usepackage{braket}
\usepackage{caption}
\usepackage[colorlinks,linkcolor=blue]{hyperref}

\theoremstyle{definition}
\theoremstyle{definition}
\begin{document}
	\title{\bf{Characterization of exact one-query quantum algorithms (ii): for partial functions }}
	\author{Zekun Ye$^{1}$, Lvzhou Li$^{1, 2, }$ }\thanks{lilvzh@mail.sysu.edu.cn.}%
	
	\affiliation{%
		$^1$ Institute of Computer Science Theory, School of Data and Computer Science, Sun Yat-Sen University, Guangzhou 510006, China
	}%

	\affiliation{%
		$^2$ Ministry of Education Key Laboratory of Machine Intelligence and Advanced Computing (Sun Yat-sen University), Guangzhou  {\rm 510006}, China
	}%
	
	
	\begin{abstract} The
 query model  (or black-box model) has attracted much attention  from the communities  of both classical and quantum computing. Usually, quantum advantages are revealed by presenting a quantum   algorithm that has  a better query complexity than its classical counterpart. For example, the well-known quantum algorithms including Deutsch-Jozsa algorithm, Simon algorithm and  Grover algorithm  all show  a considerable advantage of quantum computing from the viewpoint of query complexity. Recently we have considered in (Phys. Rev. A. {\bf 101}, 02232 (2020)) the problem: what
functions can be computed by an exact one-query quantum algorithm?
 This problem has been addressed for  total Boolean functions but still open for partial Boolean functions. 
 Thus, in this paper we continue to characterize the computational power of exact one-query quantum algorithms for partial Boolean functions by giving several necessary and sufficient conditions. By these conditions,  we construct some new functions that can be computed exactly by  one-query quantum  algorithms but have essential difference from the already known ones.	Note that before our work, the known functions that can be computed by exact one-query quantum  algorithms are all symmetric functions, whereas the ones constructed in this papers are generally asymmetric.
 \end{abstract}

	\maketitle

	\section{ Introduction}

	The  decision tree model  has been well studied in classical computing,  and focuses on problems such as the following: given a Boolean function $f:\{0,1\}^n\rightarrow \{0,1\}$,  how can we make  as few queries as possible  to the bits of $x$ in order to output the value of  $f(x)$? Quantum analog, called the quantum query model, has also attracted much attention in recent years \cite{Buhrman2002Complexity}. The implementation procedure of a quantum query model is  a  quantum query algorithm, which can be roughly described as follows:  it starts with a fixed state $|\psi_0\rangle$, and then performs the sequence of operations $U_0, O_x, U_1,  \ldots, O_x,U_t$, where $U_i$'s are unitary operators that do not depend on the input $x$ but the query $O_x$  does. This leads to the final state $ |\psi_x\rangle=U_tO_xU_{t-1}\cdots U_1O_xU_0|\psi_o\rangle$. The result is obtained by measuring the final state  $ |\psi_x\rangle$.

	The quantum query model can be discussed in two main settings: the exact setting and the bounded-error setting. A quantum query algorithm is said to compute a function $f$ exactly, if its output equals $f(x)$ with probability 1, for all inputs $x$. In this case, the algorithm is called an {\it exact quantum algorithm}. It is said to {\it compute $f$ with bounded error},  if its output equals $f(x)$ with a  probability greater than a constant, for all inputs $x$. Roughly speaking, the query complexity of a function $f$  is the number of queries  that an  optimal (classical or quantum) algorithm should make in the worst case to compute $f$. The classical deterministic query complexity of  $f$ is denoted by $D(f)$, and the quantum query complexity in the  exact setting is denoted by $Q_E(f)$. In this paper, we focus on quantum query algorithms in the exact setting  \cite{Chen2020Char,Ambainis2013Superlinear, Ambainis2015Exact, Deutsch1992Rapid, Midrijanis2004Exact, Mihara2003Deterministic, He2018Exact, Montanaro2015On, Ambainis2013Exact, Ambainis2017Exact, Cai2018Optimal, Ambainis2016Superlinear, Cleve1998Quantum, Vasilieva2006Computing, Aaronson2016separations, Ambainis2017separations, Farhi1998alimit, Hayes2002quantum, Mischenko2015quantum, Braunstein2007exact, brassard1997exact, Qiu2018Generalized, Qiu2020Revisiting}, where quantum advantages were shown by comparing $Q_E(f)$ and $D(f)$. For total Boolean functions, Beals et al. \cite{beals2001quantum} showed that exact quantum query algorithms can only achieve  polynomial speed-up over classical counterparts.  At the same time, Ambainis et al. \cite{Ambainis2015Exact} proved that exact quantum algorithms have advantages for almost all Boolean functions.  However, the biggest gap between $Q_E(f)$ and $D(f)$ is only a factor of 2 and is achieved by Deutsch algorithm for a long time. In 2013,  a breakthrough result  was obtained by Ambainis, showing   the first total Boolean function for which exact quantum algorithms have superlinear advantage over  classical deterministic algorithms \cite{Ambainis2013Superlinear}. Moveover, Ambainis  \cite{Ambainis2016Superlinear} improved this result and presented a nearly quadratic separation in 2016.
	
 For partial functions (promise problems),  exponential separations between exact quantum and classical deterministic query complexity were obtained in several papers \cite{Deutsch1992Rapid,brassard1997exact,Mihara2003Deterministic, Cai2018Optimal}. A typical example is  Deutsch-Jozsa algorithm \cite{Deutsch1992Rapid}. In addition, some work showed an exponential separation between quantum and randomized query complexity in the bound-error setting, such as Simon algorithm \cite{Simon1997On} and Shor algorithm \cite{Shor1994Discrete}.
 Recently, Childs and Wang \cite{Childs2020Can} proved that there is at most polynomial quantum speedup for (partial) graph property problems in the adjacency matrix model. On the contrary, in the adjacency list model for bounded-degree graphs, they exhibited a promise problem that shows an exponential separation between the randomized and quantum query complexities. Moreover, a series of work showed that for partial Boolean functions on $N$ variables, the quantum query complexity could be exponentially smaller (or even less) than the randomized query complexity. 
 Aaronson and Ambainis \cite{Aaronson2018Forrelation} proposed a promise problem called Forrelation and showed that this problem can be solved using 1 quantum query with bounded error, yet any randomized algorithm needs  $\widetilde{\Omega}(\sqrt{N})$ queries. They also showed that this separation is essentially optimal: any $t$-query quantum algorithm can be simulated by an $O(N^{1-1/2t})$-query randomized algorithm. Tal \cite{Tal2020Towards} gave a $O(1)$ vs. $N^{2/3-\epsilon}$ separation between the quantum and randomized query complexities of partial Boolean functions by a variant of $k$-fold Forrelation problem. Furthermore, for any positive integer $k$, Sherstov et al. \cite{Sherstov2020An} obtained a partial function on $N$ bits that has bounded-error quantum query complexity at most $\lceil k/2 \rceil$ and randomized query complexity $\tilde{\Omega}(N^{1-1/k})$. This separation of bounded-error quantum versus randomized query complexity is best possible, by the results of Aaronson and Ambainis \cite{Aaronson2018Forrelation}.

Recently, characterization of one-query quantum algorithms (that can make only one query) has received some attention \cite{	Aaronson2016Polynomials, Arunachalam2019Quantum,		Chen2020Char,Qiu2020Revisiting}.  Aarsonson et al. \cite{Aaronson2016Polynomials} and Arunachalam et al. \cite{Arunachalam2019Quantum} have presented a complete characterization of the Boolean functions that can be computed by a one-query quantum algorithm in the bounded-error setting,  but their results are not applicable to the exact case. 
We  considered  the problem of what Boolean functions can be computed by exact one-query quantum algorithms, by  proving that  a total  Boolean function $f:\{0,1\}^n \rightarrow \{0,1\}$  can be  computed by an exact one-query quantum algorithm if and only if $f(x)=x_{i_1}$ or $f(x)=x_{i_1} \oplus x_{i_2} $ (up to isomorphism) \cite{Chen2020Char}. However, this does not hold for partial functions.   On the one hand, it has been known that any symmetric partial Boolean function $f$ has $Q_E(f) = 1$ if and only if $f$ can be computed by Deutsch-Jozsa algorithm \cite{Qiu2020Revisiting}. On the other hand, it is unclear for the asymmetric partial functions.

All the above motivates us to consider the following question: what partial Boolean functions can be computed by an exact one-query quantum algorithm? We answer this question by giving some necessary and sufficient conditions.  Furthermore, by these conditions we obtain and discuss some new representative function examples that can be computed  by exact one-query quantum algorithms. These examples generalize the known function cases and fill a gap in the discussion of exact one-query asymmetric functions. Especially, some cases have essential difference from the known functions that can be computed by  Deutsch-Jozsa algorithm. Note that recently Xu and Qiu \cite{Xu2020Partial} have independently carried out an interesting work on the similar topic as in this paper.
	
The remainder of this paper is organized as follows. The query model and the problem we consider are given in Section \ref{Pre}. Some necessary and sufficient conditions are presented in Section \ref{Result}. The construction and discussion of new functions are shown in Section \ref{new functions}. Finally, a conclusion is made in Section \ref{Con} and some further problems are proposed.
	
\section{Preliminaries} \label{Pre}
	In this paper,  we consider  Boolean functions	$ f:\{0,1\}^n \rightarrow \{0,1\}$.  It is called a total function,  if it is defined for all $x\in\{0,1\}^n $. It is called a partial function, if it is defined on a subset $D\subset \{0,1\}^n$.  In the following, we first give an introduction about the query models, including both  classical and quantum cases, and then we describe  the problem  to be discussed.

	Given a Boolean function $ f:\{0,1\}^n \rightarrow \{0,1\}$, suppose $x=x_1x_2...x_n \in \{0,1\}^n $ is an input of $f$ and we use $x_i$ to denote its $i$-th bit. The goal of a query algorithm is to compute $f(x)$, given queries to the bits of $x$.
	
	\begin{figure}[ht]
		\centering
		\includegraphics[width=0.4\linewidth]{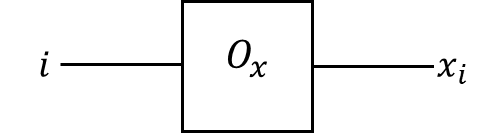}
		\caption{classical oracle}\label{CO}
		\label{classical oracle}
	\end{figure}
	
	In the classical case, the process of querying  $x$ is implemented by using the black box (called query oracle) shown in Figure \ref{CO}. We want to compute $f(x)$ by using the query oracle as few as possible. 	A classical deterministic algorithm for computing $f$  can be described by a {\it decision tree}. 
	For example, suppose that we want to use a classical deterministic algorithm to compute  $f(x)=x_1 \land (x_2 \vee x_3)$.  Then  a decision tree $T$ for that  is depicted  in Figure \ref{decision tree}. 	Given an input $x$, the tree is evaluated as follows. It starts at the root. At each node, if it is a leaf, then  its label is output as the result for $f(x)$; otherwise, it queries its label variable  $x_i$. If $x_i$ = 0, then we recursively evaluate the left subtree. Otherwise, we recursively evaluate the right subtree. The query complexity of tree $T$ denoted by $D(T)$ is its depth, and we have $D(T)=3$ in this example. Given $f$, there exist different decision trees to compute it, and the query complexity of $f$, denoted by $D(f)$, is defined as 
$D(f)=\min_TD(T).$

	\begin{figure}[H]
		\centering
		\includegraphics[width=0.3\linewidth]{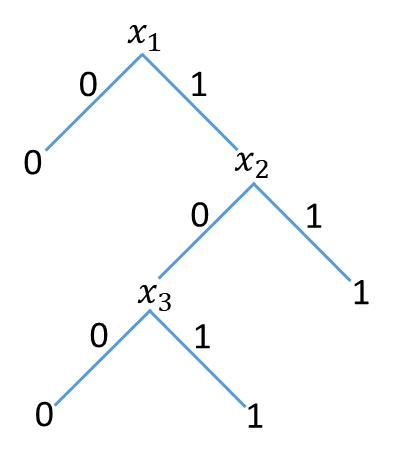}
		\caption{A decision tree $T$ for computing $f(x)=x_1 \land (x_2 \vee x_3)$}
		\label{decision tree}
	\end{figure}

	In the quantum case, we are able to query more than one bit each time due to quantum superposition. There are two equivalent query oracles: the bit flip oracle $\hat{O}_x$ defined by $\hat{O}_x|i,b\rangle = |i,b \oplus x_i\rangle$ for $i \in \{1,...,n\}$, $b \in \{0,1\}$; the phase oracle $O_x$ defined by $O_x\ket{i} = (-1)^{x_i}\ket{i}$ for all $i \in \{0,...,n\}$, where  $x_0=0$. In this paper, we use the phase oracle for convenience. A $T$-query quantum algorithm can be seen as a sequence of unitaries $U_TO_xU_{T-1}O_x...O_xU_0$, where $U_i$'s are fixed unitaries and $O_x$ depends on $x$. The process of computation is as follows:
	\begin{itemize}		
		\item[ (1)]  Start with an initial state $|\psi_{0}\rangle$.
		
		\item[ (2)] Perform the operators $U_0, O_x, U_1, O_x...U_T$ in sequence, and then we obtain the state $|\psi_x\rangle=U_TO_xU_{T-1}O_x...U_0|\psi_{0}\rangle$.
		
		\item[ (3)] Measure $|\psi_x\rangle$ with a $0-1$ positive operator-valued measurement  \cite{Nielsen2002Quantum}. The measurement result
		is regarded as the output of the algorithm.
	\end{itemize}

	In the above,  we use $r(x)$ to denote the measurement result of $|\psi_x\rangle$. Let $P[\mathcal{A}]$ denote the probability that event $\mathcal{A}$ occurs.  If it satisfies:
	
	\begin{equation}  \nonumber
	\forall x , P[ r(x)=f(x) ]  \geq 1-\epsilon   , 
	\end{equation}
	where $\epsilon < \frac{1}{2} $, then the quantum query algorithm is said to compute $f(x)$ with bounded error $\epsilon$. If it satisfies:
	\begin{equation} \nonumber
	\forall x, P[ r(x)=f(x) ]  =1,
	\end{equation} 
	then it is said to compute $f(x)$ exactly, and the algorithm is called an exact quantum algorithm.
	The exact quantum  query complexity of   $f$, denoted by $Q_E(f)$, is the minimum number of queries that a quantum query algorithm needs to compute $f$.
	The gap between $D(f)$ and $Q_E(f)$ is usually used to exhibit quantum advantages.

	In this paper, we want to characterize those partial Boolean functions $f$ that satisfy $Q_E(f)=1$. In other words, we consider this problem: what partial Boolean functions $f$ can be computed by an exact one-query quantum algorithm?  

	\section{Necessary and sufficient conditions} \label{Result}
First, we give some notation. For $x\in\{0, 1\}^n$, we  always associate it with one bit $x_0 = 0$, and let  $\ket{x} = \sum_i x_i|i\rangle$ be an $n+1$ dimension vector.
Define $x'$ by $x'_i = (-1)^{x_i}$ for any $i$. For a non-negative diagonal matrix $D$, let $|x_D\rangle = \sqrt{D}|x'\rangle$. If $tr(D) = 1$, then $|x_D\rangle$ is a unit vector. 
   
    We given some necessary and sufficient conditions as follows.
	
	\begin{theorem}
		\label{theorem1}
		The following statements  are equivalent to each other:
	\begin{itemize}
	\item 	[(a)] $f:\{0,1\}^n \rightarrow \{0,1\}$ can be computed by an exact one-query quantum algorithm. \label{Th(a)}
		
	\item 	[(b)] There exist a set of non-negative coefficients $\{c_i\}$ with $\sum_{i=0}^n c_i = 1$, such that $f(x) \neq f(y)$ implies $\sum_{i \in S}c_i = \frac{1}{2}$, where $ S = \{i|x_i \neq y_i\}$.

	\item 	[(c)]  There exists a non-negative diagonal matrix $D$ with $tr(D) = 1$, such that $f(x) \neq f(y)$ implies $\langle x_D|y_D \rangle = 0$.

	\item 	[(d)]  There exists a project operation $P$ and a non-negative diagonal matrix $D$ with $tr(D) = 1$, such that $f(x) = \langle x_D |P| x_D \rangle$.
\end{itemize}

	\end{theorem}

	\begin{proof} $(a)\Rightarrow (b)$.
	Suppose we have an exact one-query quantum algorithm. let $|\psi_x\rangle = U_1 O_x U_0|\psi_0\rangle $. 
	 Since the algorithm is allowed to use auxiliary space, we assume 
	\begin{equation}
	\begin{aligned}
	|\phi_x\rangle &= O_x U_0|\psi_0\rangle \\
	&= O_x \sum_{i,j}\alpha_{ij}\ket{i}\ket{j} \\
	&= \sum_{i,j}\alpha_{ij}(-1)^{x_i}\ket{i}\ket{j}.
	\end{aligned}
	\end{equation}
	The assumption that $f$ can be computed exactly implies that $|\psi_x\rangle$ and $|\psi_y\rangle$ can be perfectly distinguished. Thus, the two states are orthogonal, i.e., $ \langle \psi_x|\psi_y\rangle=0$. Since unitary operators don't change the orthogonality between two states, equivalently, there is $\langle\phi_x|\phi_y\rangle = 0 \label{eq:cond}$, which means 
	\begin{equation}
	\begin{aligned}
	\langle\phi_x|\phi_y\rangle &= \sum_{ij}|\alpha_{ij}|^2(-1)^{x_i \oplus y_i} \\
	&= \sum_{i \notin S}\sum_j|\alpha_{ij}|^2- \sum_{i \in S}\sum_j |\alpha_{ij}|^2 = 0.
	\end{aligned}
	\end{equation}
	And because $\sum_{ij} |\alpha_{ij}|^2 = 1$, we have $\sum_{i \in S}\sum_j |\alpha_{ij}|^2 = 1/2.$  Now we construct the coefficients $\{c_i\}$ in item (b). For $i \in \{0,...,n\}$, let $c_i =\sum_j |\alpha_{ij}|^2$. Then we have
\begin{equation}
\sum_{i \in S}c_i =  \sum_{i \in S}\sum_j |\alpha_{ij}|^2 = \frac{1}{2}.
\end{equation}
As a result, we get a set of feasible coefficients $\{c_i\}.$
	
 $(b)\Rightarrow (c)$.	
	We define diagonal matrix $D$ by $D_{ii} = c_{i}$ for any $i$. 
	If $f(x) \neq f(y)$, then 
\begin{equation}
\begin{aligned}
\langle x_D |y_D\rangle 
& = \langle x'|D|y'\rangle \\
& =\sum_i c_{i} x'_i y'_i \\
&= \sum_{i:x_i = y_i}c_{i}- \sum_{i:x_i \neq y_i}c_{i} \\ 
&= \sum_{i \notin S}c_i-\sum_{i \in S}c_i =0.
\end{aligned}
\end{equation}

 $(c)\Rightarrow (d)$.
	By item (c), if $f(x) \neq f(y)$, then $\langle x_D|y_D \rangle =0$. Let $X = \{|x_D\rangle|f(x)=1\}$, $Y = \{|y_D\rangle|f(y) = 0\}$. Then $X \bot Y$. Select an orthonormal basis $\{|v_i\rangle\}$ of  $Span\{X\}$. 
	Let $g(x)=\sum_{i} \langle v_i|x_D\rangle^2$. If $f(x)=1$, then $|x_D\rangle \in Span\{X\}$, thus we have $g(x) = \| |x_D \rangle \|^2= 1$. If $f(x)=0$, then for any $i$, $\langle v_i|x_D \rangle = 0$, thus $g(x)=0$. As a result, we have $f(x)=g(x)$. 
	Let $P = \sum_{i} |v_i \rangle \langle v_i|$. 
	Then 
	\begin{equation}
	\begin{aligned}
	f(x) &= g(x) = \sum_{i} \langle v_i| x_D \rangle^2 \\
	&= \sum_i \langle x_D |v_i\rangle \langle v_i| x_D\rangle \\
	&= \langle x_D|(\sum_i|v_i\rangle \langle v_i|)| x_D\rangle \\
	&=\langle x_D|P|x_D\rangle.
	\end{aligned}
	\end{equation}

 $(d)\Rightarrow (a)$.
	Suppose 
	$f(x) = \langle x_D|P|x_D \rangle$. We give Algorithm \ref{algo1} to compute $f$ as follows, which uses only one quantum query.
\begin{algorithm}[H]
			\caption{One-query algorithm}
			\label{algo1}
			\textbf{Input:} $n$-bit Boolean string $x$. 
				
				\textbf{Output:} $f(x)$.
		
			\textbf{Procedure:}
			
			\begin{enumerate}
				\item Prepare the initial state $\sum_i \sqrt{D}_{ii} | i\rangle$.\label{procedure1}
				\item Perform the operation $O_x$ to the initial state and then obtain the state:
                \begin{align*}
                O_x\sum_i \sqrt{D}_{ii} | i\rangle
                &= \sum_i \sqrt{D}_{ii}  (-1)^{x_i}| i\rangle \\
                &=\sum_i \sqrt{D}_{ii} x_i' | i\rangle \\
                &= \sqrt{D}| x'\rangle \\
                &= |x_D \rangle.
                \end{align*}
				\item Measure the  register by measurement  basis $\{I-P,P\}$. If the result associated with $P$ is measured, then return $1$; else, return $0$. 
			\end{enumerate}
		
			\end{algorithm}
		One can see that the probability of output $1$ is $\langle x_D|P| x_D \rangle = f(x)$, and the probability of output $0$ is $1-\langle x_D|P|x_D\rangle = 1-f(x)$. Thus, the algorithm always outputs correct results.
		


\end{proof}
In summary, statements $(a)-(d)$ are equivalent to each other, which may offer a deeper insight into the problem of what Boolean functions can be computed by an exact one-query quantum algorithm. By Theorem \ref{theorem1}, one-query function $f$ satisfies $f(x) =  \sum_i \langle v_i|x_D \rangle^2$. 
Since $\langle v_i|x_D \rangle = \langle v_i|\sqrt{D}|x' \rangle$ and $x_i' = (-1)^{x_i} = 2x_i-1$, the degree of $\langle v_i|x_D \rangle$ is 1. Thus, the expression of $f$ in statement $(d)$ satisfies that $deg(f) \le 2$, which meets the conclusion from the polynomial method \cite{Buhrman2002Complexity}.


\section{New representative functions with one quantum query} \label{new functions}




In the following, we call a function that can be computed by an exact one-query quantum algorithm as a {\it one-query function} for the sake of simplicity. Here we find some new one-query functions by Theorem \ref{theorem1}. 
These examples help us understand the power of exact one-query algorithm better. 
To begin with, we list all known one-query functions (up to isomorphism) \cite{Deutsch1992Rapid, Qiu2020Revisiting, He2018Exact}, which can be computed by  Deusch-Josza algorithm: 

i) Deusch-Josza function $f_1:\{0,1\}^n \rightarrow \{0,1\}$: 
$$f_1(x)=
\begin{cases}
0,&  |x|=n/2\\
1,&  |x| = 0 \text{ or } n
\end{cases}.$$

ii) Symmetric  function $f_2:\{0,1\}^n \rightarrow \{0,1\}$: 
$$f_2(x)=
\begin{cases}
0,&  |x|=c \ (c \ge \lceil n/2 \rceil)\\
1,&  |x| = 0
\end{cases}.$$

Next, we construct some new functions inspired by Theorem \ref{theorem1} as follows. It is easy to check these functions satisfy the statement $(b)$ (see Table \ref{table2}).

iii) For any set of non-negative coefficients $\{c_i\}$ satisfying $\sum_{i=0}^n c_i = 1$, there exists a quasi-symmetric function $f_3:\{0,1\}^n \rightarrow \{0,1\}$: 
$$f_3(x)=
\begin{cases}
0,&  \hat{x}=1/2\\
1,&  \hat{x}=0 \text{ or } 1
\end{cases},$$
where $\hat{x} = \sum_{i=0}^n c_i x_i$.

iv) Function $f_4: \{0,1\}^4 \rightarrow \{0,1\}$: 
$$f_4(x)=
\begin{cases}
0,&  x=0000,0011,1100,1111\\
1,&   x=0101,0110,1001,1010
\end{cases}.$$

v) Function $f_5:\{0,1\}^{4n} \rightarrow \{0,1\}$ as Table \ref{table1}. 

\begin{table}
    \centering
    \begin{tabular}{c|c}
    \hline
         $x:f(x) = 1$ & $x:f(x) =0$ \\
         \hline
              $\underbrace{0...0}_{4n}$ \\
        $\underbrace{1...1}_{4n}$ &$|x|=2n,$\\
         $\underbrace{0...0}_{2n}\underbrace{1...1}_{2n}$  &$\sum_{i=1}^{2n}x_i = n$, and \\
        $\underbrace{1...1}_{2n}\underbrace{0...0}_{2n}$ & $\sum_{i=1}^n x_i + \sum_{i=2n+1}^{3n} = n.$\\ 
          $\underbrace{0...0}_{n}\underbrace{1...1}_{n}\underbrace{0...0}_{n}\underbrace{1...1}_{n}$ \\
          $\underbrace{1...1}_{n}\underbrace{0...0}_{n}\underbrace{1...1}_{n}\underbrace{0...0}_{n}$  \\
    \end{tabular}
    \caption{Function $f_5$}
    \label{table1}
\end{table}
  


As mentioned earlier, a one-query function $f$ has the expression: $f(x) =  \sum_i \langle v_i|\sqrt{D}|x' \rangle^2$. If a vector set $\{w_i\}$ satisfies $|v_i\rangle = \sqrt{D}|w_i\rangle$, then $f(x) =  \sum_i \langle w_i|D|x' \rangle^2$. Since $D_{ii} = c_i$ for any $i$, we give the corresponding $\{c_i\}$ and $\{w_i\}$ for above functions as Table \ref{table2}. For any above function and $i$, $w_i$ is an $n$-bit Boolean string. By $\{c_i\}$ and $\{w_i\}$, we can obtain $D$  and $P$  easily, and then use Algorithm \ref{algo1} to compute these functions.   
\begin{table}
    \centering
    \begin{tabular}{c|c|c}
        $f$ & $\{c_i\}$ & $\{w_i\}$ \\
        \hline
        $f_1$ & $c_{0} = 0, c_{i} = \frac{1}{n} (\forall i>0)$ & $w_1 = \underbrace{1...1}_{n}$ \\
        $f_2$ & $c_{0} = \frac{2c-n}{2c}, c_{i} = \frac{1}{2c} (\forall i>0)$ & $w_1 = \underbrace{1...1}_{n}$ \\
        $f_3$ & $\sum_i c_i =1$ & $w_1 = \underbrace{1...1}_{n}$ \\
        $f_4$ & $c_{i} = 0 (i \le 2), c_{i}=1/2 (i > 2)$ & $w_1 = 0011$ \\
        \hline
        \multirow{6}{*}{$f_5$} & \multirow{6}{*}{$c_{0} = 0, c_{i} = \frac{1}{4n} (\forall i>0)$} & $w_1 = \underbrace{1...1}_{4n}$ \\
        &&$w_2 = \underbrace{1...1}_{2n}\underbrace{0...0}_{2n}$ \\
        &&$w_3 = \underbrace{1...1}_{n}\underbrace{0...0}_{n}\underbrace{1...1}_{n}\underbrace{0...0}_{n}$ \\
        
    \end{tabular}
    \caption{function examples and their corresponding $\{c_i\}$ and $\{w_i\}$}
    \label{table2}
\end{table}

\subsection{Characteristics of new functions}
Different from the already known one-query functions, the new cases are generally asymmetric functions. In this way, our results fill a gap in the discussion of exact one-query asymmetric functions. Next, we describe the characteristics of these functions respectively. 

i) Actually, $f_1$ and $f_2$ are both  instances of $f_3$. If $c_0 = 0$ and $c_1 = \dots =c_n$, then $f_3$ degenerates into $f_1$; if  $c_0 \neq 0$ and $c_1 = \dots =c_n$, then $f_3$ degenerates into  $f_2$. Otherwise, $f_3$ is asymmetric. In this way, $f_3$ is a generalization of $f_1$ and $f_2$. It not only can represent all symmetric one-query functions, but also represents a broad class of asymmetric one-query functions.

ii) Unlike previous functions, $f_4$ is the first function have the following property: any function $f$ isomorphic to $f_4$ does not satisfy that $f^{-1}(1) \subseteq \{x:|x|=0$ or $n\}$. It is also worth mentioning that $f_4$ and $f_1$ have the same domain when $n=4$.

iii) 
$f_5$ is the first proper example represented as the sum of square form: $\sum_i \langle w_i|D|x' \rangle^2$, whereas all previous examples can be represented as single square expressions: $\langle w_1|D|x' \rangle^2$ (see Table \ref{table2}). 
As a result, $f_5$ is more general and representative than previous cases in all one-query functions.

	\section{Conclusion} \label{Con}
	In this paper we have presented several  necessary and sufficient conditions for partial Boolean functions being computed by exact one-query quantum algorithms. By these conditions, we have obtained  some new function examples which are asymmetric and have essential difference from the already known symmetric one-query functions.
We hope our results are helpful for further discussion	of the power of exact $k$-query quantum algorithms. Figuring out this problem is useful for  understanding in depth quantum query algorithms and inspiring us to find more problems with quantum advantages.
	
	
	\bibliographystyle{unsrt}

\end{document}